\documentclass[journal,twoside,9pt]{IEEEtran}

\usepackage{mathrsfs}
\usepackage{amsfonts}
\usepackage{amsmath,amssymb}
\usepackage{graphicx}
\usepackage{float}
\usepackage[dvips]{epsfig}
\usepackage{subfigure}
\usepackage{balance}

\usepackage{tikz}
\usepackage{tkz-graph}

\usepackage{multicol}

\usepackage{xcolor}

\def\eqref#1{(\ref{#1})}

\usepackage{amsthm}

\theoremstyle{definition}

\newtheorem{assumption}{Assumption}
\newtheorem{theorem}{Theorem}
\newtheorem{lemma}{Lemma}

\newtheorem{remark}{Remark}

\begin{document}

%
\title{Distributed Optimal Control with Recovered Robustness for Uncertain Network Systems: A Complementary Design Approach}

\author{Zhongkui Li, Junjie Jiao, and Xiang Chen
\thanks{This work was supported in part by the National Natural Science Foundation of China under grants 619730065 and U1713223.}
\thanks{Z. Li is with the State Key Laboratory for Turbulence and Complex Systems, Department of Mechanics and Engineering Science, College of Engineering, Peking University, Beijing 100871, China. E-mail: {\tt zhongkli@pku.edu.cn}}
\thanks{J. Jiao is with the Chair of Information-oriented Control, Department of Electrical and Computer Engineering, Technical University of Munich, 80333, Munich, Germany. E-mail: {\tt junjie.jiao@tum.de}}
\thanks{X. Chen is with the Department of Electrical and Computer Engineering, University of Windsor, Canada.
E-mail: {\tt xchen@uwindsor.ca}}
}

\IEEEtitleabstractindextext{%
\begin{abstract}
This paper considers the distributed robust suboptimal consensus control problem of linear multi-agent systems,  with both   $H_2$ and $H_\infty$ performance requirements. A novel two-step complementary  design approach is proposed. In the first step, a distributed control law is designed for the nominal multi-agent system to achieve consensus with a prescribed $H_2$ performance. In the second step, an extra control input, depending on some carefully chosen residual signals indicating the modeling mismatch, is designed to complement the $H_2$ performance by providing robustness guarantee in terms of $H_\infty$ requirement with respect to disturbances or uncertainties.
The proposed complementary design approach provides an additional degree of freedom for design,  having two separate controls to deal with the $H_2$ performance and the robustness of consensus, respectively. Thereby, it    does not need to make { much} trade-off, and can be expected to be much less conservative than the trade-off design such as the mixed $H_2/H_\infty$ control method.
Besides,  this complementary approach will recover the achievable $H_2$ performance when external disturbances or uncertainties do not exist.
The effectiveness of the theoretical results and the advantages of the complementary approach are validated via numerical simulations.
\end{abstract}

\begin{IEEEkeywords}
Robust control, cooperative control, consensus, distributed control, optimal control, $H_\infty$ control, $H_2$ control.
\end{IEEEkeywords}}

\maketitle

\IEEEdisplaynontitleabstractindextext

\IEEEpeerreviewmaketitle

\section{Introduction}
Optimality and robustness are two main issues and missions in the feedback control theory \cite{zhou1998essentials}. The optimality requires an optimal or suboptimal controller to ensure that the closed-loop system satisfies certain predefined performance criteria, with the linear quadratic regulator (LQR) and the linear quadratic Gaussian (LQG)/$H_2$ problems as typical examples.
The robustness, on the other hand, characterizes the property that a system still works well in the presence of external disturbances and model uncertainties, which can be addressed in the framework of $H_\infty$ control and $\mu$ synthesis.  For   multi-agent network systems, the optimality and robustness problems encounter new inherent challenges, since the control laws need to be distributed in the sense that only local information between neighboring agents can be utilized and meanwhile the control laws are subject to structural constraints imposed by the network topology \cite{li2014cooperative}.

In the last decades, many advances have been reported on the optimality and robustness issues of network systems. Distributed optimal and suboptimal LQR problems of multi-agent systems were addressed in \cite{borrelli2008distributed,nguyen2015sub,movric2013cooperative,jiao2019distributed,jiao2019suboptimality}. The coherence and centrality of multi-agent networks were formulated and discussed in \cite{bamieh2012coherence,siami2017centrality,siami2017new,young2010robustness,patterson2014consensus,liumiao2020}, through studying the $H_2$ norms of the networks from the white noises to the performance variables.  The $H_2$ consensus problems of multi-agent systems were investigated in \cite{li2010hinf,JIAO2020109164}, where distributed state and output feedback consensus protocols were designed to satisfy prescribed $H_2$ performance indices.  Disturbance attenuation problems of linear consensus networks were studied from the $H_\infty$ control perspective in \cite{li2009h,oh2014disturbance,li2010hinf,massioni2009distributed,saboori2014h}. For multi-agent systems, both the agent dynamics and the interactions among neighboring agents could be perturbed by uncertainties. The robust synchronization and consensus problems of multi-agent systems whose agent dynamics are subject to multiplicative and coprime factor uncertainties were investigated in \cite{trentelman2013robust,lixw2018robust,jongsma2016robust}.  The authors in \cite{li2019robust,ma2014mean,li2014multi,xu2016consensusability} considered the robustness of multi-agent consensus in the sense of mean square and almost sure stabilities for the cases where the communication channels among agents are subject to multiplicative stochastic uncertainties. Robust consensus over deterministic uncertain network graphs was studied in \cite{li2017robust,7286773}.

In this paper, we will address the distributed robust optimal consensus problems for linear multi-agent systems, taking into account both the optimality and robustness at the same time.
It is well-known that there is an intrinsic conflict between optimality and robustness in the standard feedback framework \cite{zhou1998essentials,zhou2001new}.
Therefore, in the case of multi-objective design, e.g., the mixed $H_2/H_\infty$ control  \cite{zhou1994mixed} and the $H_\infty$ Guassian control \cite{chen2001multiobjective}, a trade-off has to be made between the achievable optimal performance and robustness \cite{zhou2001new}. These trade-off design approaches suffer  from a fundamental drawback of severe conservativeness,  because a single controller is developed to address the conflicting requirements simultaneously. For example, the mixed $H_2/H_\infty$ control is generally worse than the $H_\infty$ control in terms of the robustness and  worse than the LQG control in terms of the optimality \cite{chen2001multiobjective}.

The objective of this paper is to present a novel {\em non-trade-off} design approach to the robust optimal consensus problems for linear multi-agent systems. Motivated by the structure in \cite{zhou2001new}, a new design paradigm
is proposed in \cite{chen2019revisit}, consisting of an LQG control designed for the nominal plant  and an   operator $Q$ as a separate degree of freedom. The  operator $Q$ provides an extra control action to recover the robustness   performance for the closed-loop system. This new paradigm is shown to be able to avoid trade-off and to reduce the conflict between the robustness and achievable  suboptimality/optimality. It should be noted that the design paradigms in \cite{zhou2001new,chen2019revisit} are applicable to only single-agent systems. Novel non-trade-off design schemes for multi-agent systems have not witnessed  significant progress so far, due to the severe difficulties caused by the requirement of distributed control and structural uncertainties and constraints imposed by the network graphs.

In this paper, motivated by \cite{zhou2001new,chen2019revisit}, we propose a {\em complementary  design approach} to the $H_2$ and $H_\infty$ consensus problem of linear multi-agent systems. This complementary  approach consists of two steps. In the first step, we design a distributed control law for the nominal multi-agent  system, without considering the disturbances or uncertainties, to ensure that consensus is achieved with a prescribed $H_2$ performance. In the second step, a separate control input, activated by some carefully chosen residual signals indicating the modeling mismatch, will be designed to ensure robustness in terms of the $H_\infty$ requirement.
 Two cases are considered, namely, the case that relative outputs or  absolute outputs   of neighboring agents are available. Suitable residual signals are chosen for both  cases.
For  the first case with relative outputs, the residual signal is defined as the stacked error between the actual relative outputs of neighboring agents and their observed ones given by a distributed observer.
For the   second case with absolute outputs, the residual signal quantifies the differences between the actual agent dynamics and the ideal  agent dynamics. A distinct feature of this complementary approach is that the design of $H_2$ consensus control in this first step is independent of the second step and the extra control action in the second step will complement the $H_2$ performance by providing a robustness guarantee  with respect to disturbances or uncertainties.

Compared to the trade-off approach, e.g.,  the mixed $H_2/H_\infty$ control design, the proposed complementary design approach has at least two main advantages. Firstly, since   the extra control
provides an additional degree of freedom for design, the complementary approach has two separate controls to deal with the $H_2$ performance and the robustness of consensus,  respectively. Thereby, this approach does not need to make { much} trade-off, and can be expected to be much less conservative than the trade-off approach where one control tackles two conflicting performances.
Secondly, the control action of the second step is proportional to the residual signal which quantifies the modeling mismatch level, thereby having some online ``adaptivity" with respect to modeling errors.
This complementary approach will yield the same achievable $H_2$ performance when modeling mismatches do not exist.  By contrast, the trade-off approach always considers the {\it a priori} worst case and still yield the same conservative performance even when disturbances or uncertainties do not exist.

The remainder of this paper is organized as follows. Some mathematical preliminaries including graph theory and results on $H_2$ and $H_\infty$ performances are summarized in Section \ref{sec_preliminaries}.  The $H_2$ and $H_\infty$ consensus problem is formulated in Section \ref{sec_problem}. A two-step complementary design approach is proposed for the $H_2$ and $H_\infty$ consensus problem in Section  \ref{sec_twostep}.
A  simulation example that illustrates the proposed theoretical results are presented in Section \ref{sec_simu}.
Finally, conclusions are given in Section \ref{sec_conclusions}.

\section{Mathematical Preliminaries}\label{sec_preliminaries}

\subsection{Notations} The notations used in this paper is standard. $\mathbf{R}^{n\times m}$ denotes the set of  $n\times m$
real matrices, $I$ represents the identity matrix of appropriate dimension, and $\mathbf{1}$ denotes a column vector with all entries equal to 1. The matrix inequality $A>B$ means that $A$ and $B$ are symmetric matrices and $A-B$ is positive definite. For a square matrix $A$, $\mathrm{tr}(A)$ represents its trace.
$A\otimes B$ represents the Kronecker product of matrices $A$ and $B$.
The expectation operator is denoted by ${\mathbf{E}\{\cdot}\}$.

\subsection{Graph Theory}

The information flow among the agents can be conveniently modeled by a graph. An undirected graph is defined by  $\mathcal{G}=(\mathcal{V},\mathcal{E})$,
where $\mathcal{V}=\{1,\cdots,N\}$ is the set of nodes (each node represents an agent)
and $\mathcal{E}\subseteq \mathcal{V}\times\mathcal{V}$ denotes the set of  unordered pairs of nodes, called edges. 
An undirected graph is connected, if there exists a path between every two distinct nodes. For an undirected graph $\mathcal{G}$, its adjacency matrix, denoted by $\mathcal{A}=[a_{ij}]\in \mathbf{R}^{N\times N}$, is defined such that $a_{ii}=0$, $a_{ij}=a_{ji}=1$ if $(i,j)\in\mathcal {E}$ and $a_{ij}=0$
otherwise. The Laplacian matrix $\mathcal{L}=[\mathcal{L}_{ij}]\in \mathbf{R}^{N\times N}$ associated with
$\mathcal{G}$ is defined as $\mathcal{L}_{ii}=\sum_{j=1}^{N}a_{ij}$ and $\mathcal{L}_{ij}=-a_{ij}$, $i\neq j$.

\begin{lemma}[\cite{li2014cooperative}]\label{lem2}
For an undirected graph $\mathcal {G}$, zero is an eigenvalue of $\mathcal {L}$ with $\mathbf{1}$ as an
  eigenvector and all nonzero eigenvalues are positive.
Moreover, zero is a simple eigenvalue of $\mathcal {L}$ if
and only if $\mathcal {G}$ is connected.
\end{lemma}

\subsection{Results on $H_2$ and $H_\infty$ performances}
In this subsection, we summarize some results on $H_2$   and $H_\infty$ performances of linear systems.
 %
Consider the linear system
\begin{equation}\label{linear_sys}
	\begin{aligned}
		\dot{x} &= Ax +B w,\\
		y &= C x ,
	\end{aligned}
\end{equation}
where $x \in \mathbf{R}^n$ is the state,   $y \in \mathbf{R}^m$ is the measured output, and $w \in \mathbf{R}^q$ is the external disturbance. 

Let  $G(s)=C(sI-A)^{-1}B$ be the transfer function matrix of \eqref{linear_sys}. The $H_2$ norm of $G$  is    defined to be
	$$\| G \|_2^2=\frac{1}{2\pi}\int_{-\infty}^{\infty}{\rm{tr}} (G^\ast(j\omega)G(j\omega))d\omega.$$
We then review the following well-known result on the $H_2$    performance \cite{zhou1998essentials}.
\begin{lemma}\label{lemh2}
Let $ {\gamma}_2>0$. The following statements are
equivalent:
\begin{itemize}
\item[i)] $A$ is stable and $\|G\|_2< {\gamma}_2$.

\item[ii)]There exists   $X>0$ such that
$$AX+XA^T+BB^T<0,~\mathrm{tr}(CXC^T)< {\gamma}_2^2.$$

\item[iii)] There exist $P>0$ and $Q>0$ such that
$$\begin{aligned}
\begin{bmatrix} A^TP+PA & PB\\ B^TP & -I\end{bmatrix}<0,~\begin{bmatrix} P & C^T\\ C & Q\end{bmatrix}>0,~{\mathrm{tr}}(Q)< {\gamma}_2^2.
\end{aligned}$$
\end{itemize}
\end{lemma}

{
Next, we review the $H_\infty$    performance of \eqref{linear_sys}.
If $A$ is stable, the $H_\infty$ norm of \eqref{linear_sys} is then defined to be
\begin{equation*}
	\|G\|_\infty = {\rm sup}_{ \omega \in \mathbb{R}}\sigma(G(j\omega)),
\end{equation*}
where $\sigma(G(j\omega))$ is the maximum singular value of $G(j\omega)$.

The following lemma presents a well-known result on  the $H_\infty$    performance \cite{zhou1998essentials,duan2006robust}.}
\begin{lemma}\label{brl}
Let
$\gamma_\infty>0$.  The following
statements are equivalent:
\begin{itemize}
\item[i)] $A$ is stable and $\|G\|_\infty<\gamma_\infty$.
\item[ii)] There exists   $X>0$ such that
$$A^TX+XA + \frac{1}{\gamma_\infty^2}XB B^TX  +C^T C<0.$$
\end{itemize}
\end{lemma}
In the next section, we will formulate the problem { to be addressed  in this paper.}

%

\section{Formulation of $H_2$ and $H_\infty$ Consensus Problem}\label{sec_problem}

Consider a network of $N$ identical linear agents subject to different noises and external disturbances. The dynamics of the $i$-th
agent are described by
\begin{equation}\label{d1}
\begin{aligned}
    \dot{x}_i &=Ax_i+B_0w_{0i}+B_1w_{i}+B_2u_i,\\
    y_i &=C_2x_i+D_0w_{0i}+D_1w_{i},\quad i=1,\cdots,N,
\end{aligned}
\end{equation}
where $x_i\in\mathbf{R}^n$ is the state,
$u_i\in\mathbf{R}^{p}$ is the control input, and $y_i\in\mathbf{R}^{m_1}$ is the measurement output
of the $i$-th agent, respectively.
In \eqref{d1},  $w_i \in \mathbf{R}^{q_1}$ denotes the external disturbance signal, representing the modeling uncertainty and/or unmodeled dynamics of the $i$-th agent,   $w_{0i} \in \mathbf{R}^{q_2}$ is a white noise signal with $\mathbf{E}\{w_{0i}(t)\} = 0$ and $\mathbf{E}\{w_{0i}(t)w_{0i}(\tau)^T\} = \delta(t -\tau)I$.
The matrices $A$, $B_0$, $B_1$, $B_2$, $C_2$, $D_0$ and $D_1$ are of suitable dimensions. The pair $(A, B_2)$ is assumed to be stabilizable and the pair $(C_2, A)$ is assumed to be detectable.
The communication graph among the $N$ agents is represented by an undirected graph $\mathcal{G}$.

The agents in \eqref{d1} are said to achieve consensus if there exist control laws $u_i$ such that, given $w_{0i} =0$ and $w_i =0$, $x_i - x_j \to 0$ as $t \to \infty$ for all $i,j = 1,\ldots,N$. In this paper, output
variables $z_i$, $i=1,\cdots,N$, as defined in \eqref{pv} (see \cite{li2010hinf,li2009h,li2014cooperative} also), are adopted to quantify the consensus performance,

\begin{equation}\label{pv}
z_i=\frac{1}{N}\sum_{j=1}^N C_1(x_i-x_j), \quad i=1,\cdots,N,
\end{equation}
where $z_i\in\mathbf{R}^{m_2}$, and $C_1\in\mathbf{R}^{m_2\times n}$ is a given constant weighting matrix. Note that other candidate performance variables could also be applied, for example, those depending on the specific network topology $\mathcal{G}$ as in \cite{JIAO2020109164,JIAO2021104872}.

Let $x=[x_1^T,\cdots,x_N^T]^T$, $w_0=[w_{01}^T,\cdots,w_{0N}^T]^T$, $w=[w_{1}^T,\cdots,w_{N}^T]^T$, $u=[u_1^T,\cdots,u_N^T]^T$, $y=[y_1^T,\cdots,y_N^T]^T$, and
$z=[z_1^T,\cdots,z_N^T]^T$.
Let $T_{w z}(s)$ and $T_{w_0 z}(s)$ denote the closed-loop transfer function matrices from $w$ to $z$ and from $w_0$ to $z$, respectively, under feedback control laws $u_i$. The following main problem to be addressed in this paper can then be formulated:

\begin{itemize}
\item[]{\bf Main Problem:} For the multi-agent system in \eqref{d1} and \eqref{pv}, given constants ${\gamma}_2>0$ and ${\gamma}_\infty>0$, find feedback control laws $u_i$ such that $\|T_{w_0z}(s)\|_2< {\gamma}_2$ and $\|T_{w z}(s)\|_\infty< {\gamma}_\infty$ and the agents in \eqref{d1} achieve consensus, that is, $x_i - x_j \to 0$ as $t \to \infty$ for all $i,j = 1,\ldots,N$ if both $w_{0i} =0$ and $w_i =0$.
\end{itemize}

We shall call this formulation `$H_2$ and $H_\infty$ Consensus Problem', referring to $\|T_{w_0z}\|_2< {\gamma}_2$ as $H_2$ consensus and $\|T_{w z}\|_\infty< {\gamma}_\infty$ as $H_\infty$ consensus, respectively.

It is well-known that $H_2$ and $H_\infty$ performance are inherently conflicting \cite{zhou1998essentials}. To the best of our knowledge so far, there has been no effective solution to the consensus problem for multi-agent systems that could guarantee non-compromised $H_2$ and $H_\infty$ performance. In the present paper, we will propose a novel {\em non-trade-off complementary  design}  for obtaining distributed control laws that address the said $H_2$ and $H_\infty$ consensus problem.
The proposed design  contains two  steps.
In the first step, a distributed  control law is proposed, which achieves $H_2$ consensus for the controlled multi-agent system.
In the second step, an extra distributed  control law is designed which achieves $H_\infty$ consensus for the overall network.
In particular, we will  provide two design methods for obtaining such distributed control laws that solve the  $H_2$ and $H_\infty$   consensus problem, based on relative output feedback and absolute output feedback, respectively.

\begin{remark}
One method to solve the $H_2$ and $H_\infty$  consensus problem is the standard trade-off  mixed $H_2/H_\infty$ design \cite{ao2020robust,sheng2017output}, i.e., using one single control law such that  both performance criteria are satisfied.
However,  it is well understood that there is an intrinsic conflict between the $H_2$  performance and  $H_\infty$ robustness in the  mixed $H_2/H_\infty$ design \cite{zhou1998essentials,zhou2001new}.

Note that the complementary design can be expected to be much less conservative than the trade-off approach, as it has two separate controls to deal with the $H_2$ performance and the $H_\infty$ robustness of consensus,  respectively.
\end{remark}

\section{A Two-step Complementary  Approach to the Distributed $H_2$ and $H_\infty$ Control Problem}\label{sec_twostep}

{ In this section,  we will provide two design methods for  obtaining  distributed control laws  that  solve the $H_2$ and  $H_\infty$ consensus problem with the proposed non-trade-off complementary  approach, based on {\em relative output feedback} and   {\em absolute output feedback}, respectively.}

\subsection{Relative Output Feedback Case}\label{subsec_output_feedback}

In this subsection, we consider the case where only the {\em relative output} information of the neighboring agents is accessible to each agent. In this case, the structure of the proposed complementary  design  is depicted in Fig.  \ref{fig2}.

\begin{figure}[htbp]
\centering
\includegraphics[width=0.8\linewidth]{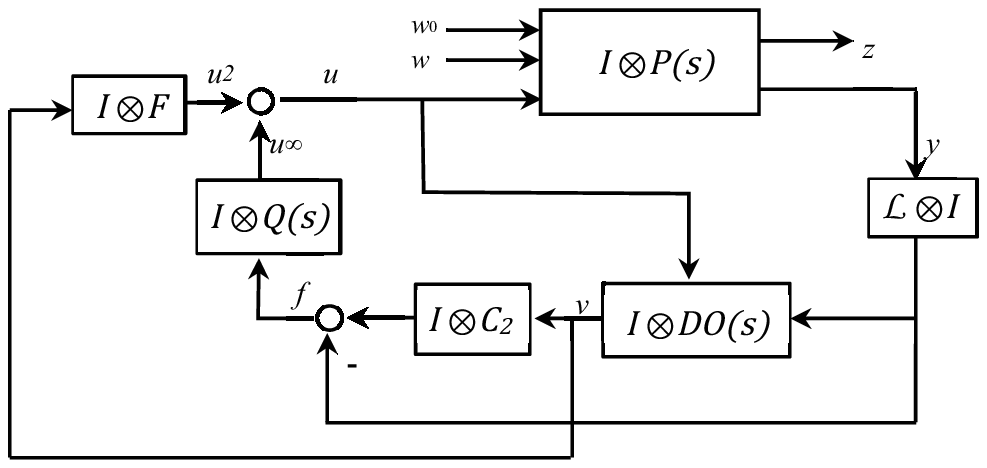}
\caption{The controller structure of the complementary  design based on  relative outputs.  In this structure, $v=[v_1^T,\cdots,v_N^T]^T$, $u_2=[u_{21}^T,\cdots,u_{2N}^T]^T$, $u_\infty=[u_{\infty1}^T,\cdots,u_{\infty_N}^T]^T$, $P(s)$ denotes  the agent  dynamics in \eqref{d1}, $DO(s)$ represents the distributed observer for each agent, with $v_i$ as the protocol state, $f=[f_1^T,\cdots,f_N^T]^T$ is the residual signal,  and $Q(s)$ is the extra controller  to compensate for $w$. The rest variables are defined as in Section \ref{sec_problem}.}\label{fig2}
\end{figure}

\subsubsection{Step One}
In the first step, we consider the   $H_2$ consensus problem     for the case with nominal agent dynamics, i.e., we consider only the noise $w_{0i}$ (without considering   external disturbances $w_i$).  Relying on the relative output information of neighboring agents, we employ the following distributed observer-based protocol \cite{trentelman2013robust,li2014cooperative}:
\begin{subequations}\label{comdiso1}
\begin{equation}\label{comdiso1a}
\begin{aligned}
\dot{v}_i  &=(A-GC_2)v_i+\sum_{j=1}^Na_{ij}\Big(B_2F( v_{i}-v_{j})+G(y_{i}-y_{j})\Big),\\
\end{aligned}
\end{equation}
\begin{equation}\label{comdiso1b}
\begin{aligned}
u_{2i} &=Fv_i,\quad i=1,\cdots,N,
\end{aligned}
\end{equation}
\end{subequations}
where $v_i\in\mathbf{R}^n$ is the protocol state, $u_{2i}$ is the input of the $i$-th agent in this step,  $F$ and $G$ are the feedback gain matrices to be designed. The coefficient  $a_{ij}$ is the $ij$-th entry of the adjacency matrix of the communication graph among the agents.

Since in the this step, we only take care of the influence of the noise $w_{0i}$ on the performance outputs $z_i$, we consider only the {\em outer loop}  in Fig.  \ref{fig2}. The control input $u_i$ of agent $i$ in this case is equal to $u_{2i}$, with $u_{\infty i}=0$.  Define the error variables
\begin{equation}\label{error_e}
	e_i\triangleq v_i-\sum_{j=1}^Na_{ij}(x_{i}-x_{j}), \quad i=1,\cdots,N.
\end{equation}
We then have
\begin{equation}\label{comdiso2}
\begin{aligned}
\dot{e}_i=(A-GC_2)e_i+(GD_0-B_0)\sum_{j=1}^Na_{ij}(w_{0i}-w_{0j}).
\end{aligned}
\end{equation}
Therefore, if $G$ is chosen such that $A-GC_2$ is Hurwitz, $v_i$ in \eqref{comdiso1} is  actually an estimate of $\sum_{j=1}^Na_{ij}(x_{i}-x_{j})$ for agent $i$.  That is, $DO(s)$ in Fig. \ref{fig2} is in fact represented by the distributed observer  in \eqref{comdiso1a}.

Denote  $v=[v_1^T,\cdots,v_N^T]^T$ and  $\xi=\begin{bmatrix} x^T & v^T\end{bmatrix}^T$. By substituting \eqref{comdiso1} into \eqref{d1},  the closed-loop network dynamics can then be written in compact form as
\begin{equation}\label{netcd1}
\begin{aligned}
    \dot{\xi}&=\mathscr{A}\xi+\mathscr{B}_0w_0, \\
    z & = \mathscr{C}_1\xi ,
\end{aligned}
\end{equation}
where
\begin{equation}\label{netcd10}
\begin{aligned}
\mathscr{A} &=\begin{bmatrix} I\otimes A  & I\otimes B_2F \\ \mathcal {L}\otimes GC_2 & I\otimes (A-GC_2)+\mathcal {L}\otimes B_2 F \end{bmatrix},\\
\mathscr{B}_0&=\begin{bmatrix} I\otimes B_0 \\ \mathcal {L}\otimes GD_0\end{bmatrix},
\mathscr{C}_1=\mathcal {M}\otimes \begin{bmatrix} C_1 & 0\end{bmatrix},
 \mathcal {M} \triangleq I - \frac{1}{N}\mathbf{1}\mathbf{1}^T.
\end{aligned}
\end{equation}


The following theorem provides a necessary and sufficient condition for the $H_2$ suboptimal consensus problem.
\begin{theorem} \label{h2c}
Assume that the graph $\mathcal {G}$ is connected. Let $\gamma_2>0$.  Then, the distributed protocol \eqref{comdiso1} { achieves $H_2$ consensus} for the network \eqref{netcd1}  if and only if the following $N-1$ subsystems
\begin{equation}\label{neticd}
\begin{aligned}
\dot{\tilde{\xi}}_i &=
   \begin{bmatrix} A  & \lambda_i B_2F \\ GC_2 & A-GC_2 +\lambda_i B_2 F\end{bmatrix}\tilde{\xi}_i +\begin{bmatrix} B_0 \\  GD_0 \end{bmatrix}\tilde{w}_{0i}, \\
    \tilde{z}_i & = \begin{bmatrix} C_1 & 0\end{bmatrix}  \tilde{\xi}_i ,\quad i=2,\cdots,N,
\end{aligned}
\end{equation}
are internally stable and $\sum_{j=2}^N\|\tilde{T}_{\tilde{w}_{0i} \tilde{z}_i}\|_2^2<\gamma_2^2$, where $\tilde{T}_{\tilde{w}_{0i} \tilde{z}_i}$ denotes the transfer function matrix of \eqref{neticd} from $\tilde{w}_{0i}$ to $\tilde{z}_i $.
\end{theorem}

\begin{proof}
The result can be proved by following similar lines in \cite{li2014cooperative,JIAO2020109164}.
The key steps are sketched here for clarity.
First, we apply the unitary transformation  $U \otimes I$ onto the dynamics of the consensus error $(\mathcal {M} \otimes I) \xi$,  where $U$ is a unitary matrix such that $U^T\mathcal {L}U=\mathrm{diag}(0,\lambda_2,\cdots,\lambda_N)$.
Note that $U^T\mathcal {M}U=\mathrm{diag}(0,1,\cdots,1)$, see e.g.,  \cite{li2014cooperative}.
Next, by observing that the $H_2$ norm is invariant under unitary transformations, we can get that the $H_2$ suboptimal consensus problem is solved if and only if the following $N-1$ subsystems
\begin{equation}\label{neticdxq}
\begin{aligned}
\dot{\bar{\xi}}_i &=
   \begin{bmatrix} A  &  B_2F \\ \lambda_iGC_2 & A-GC_2 +\lambda_i B_2 F\end{bmatrix}\bar{\xi}_i +\begin{bmatrix} B_0 \\  \lambda_iGD_0 \end{bmatrix}\tilde{w}_{0i}, \\
    \tilde{z}_i & = \begin{bmatrix} C_1 & 0\end{bmatrix}  \bar{\xi}_i ,\quad i=2,\cdots,N,
\end{aligned}
\end{equation}
are internally stable and $\sum_{j=2}^N\|\tilde{T}_{\tilde{w}_{0i} \tilde{z}_i}\|_2^2<\gamma_2^2$. Now, by letting $\tilde{\xi}_i=\left[\begin{smallmatrix}  I & 0 \\ 0 & \frac{1}{\lambda_i} I\end{smallmatrix} \right]\bar{\xi}_i$, evidently the subsystems in \eqref{neticdxq} are equivalent to those in \eqref{neticd}. 
\end{proof}

Before moving forwards, we need to make the following assumption and introduce a lemma.
\begin{assumption}\label{assum_sys_matrices}
	The system matrices in \eqref{d1} satisfy that $D_0B_0^T=0$ and $D_0D_0^T=I$.
\end{assumption}

\begin{lemma} [\cite{JIAO2020109164,haesaert2018separation}]\label{lemtr}
Suppose Assumption \ref{assum_sys_matrices} holds. Consider the $i$-th subsystem in \eqref{neticd} with $\lambda_i=1$. Let $P>0$ and $Q>0$, respectively, satisfy the following inequalities:
\begin{equation}\label{neticdxqzo1}
(A+B_2F)^TP+P(A+B_2F)+C_1^TC_1<0,
\end{equation}
\begin{equation}\label{neticdxqzo}
	AQ+QA^T-QC_2^TC_2Q+B_0B_0^T<0.
\end{equation}
If the inequality
$$\mathrm{tr}(C_2QPQC_2^T)+\mathrm{tr}(C_1QC_1^T)<\gamma^2$$
holds, then $\tilde{T}_{\tilde{w}_{0i} \tilde{z}_i}$, with $G=QC_2^T$ and $\lambda_i=1$, satisfies that $\|\tilde{T}_{\tilde{w}_{0i} \tilde{z}_i}\|_2<\gamma$.
\end{lemma}

The following theorem provides a design method for obtaining distributed protocols \eqref{comdiso1} that achieves $H_2$ suboptimal consensus.

\begin{theorem} \label{h2cdesign}
Assume that Assumption \ref{assum_sys_matrices} holds and the graph $\mathcal {G}$ is  connected. Let $\gamma_2>0$.
Let $Q>0$ be a solution to   \eqref{neticdxqzo}. Let $P>0,W>0,\tau>0$ be solutions to the following LMIs:
\begin{equation}\label{neticdxqzo2}
\begin{bmatrix}\bar{P}A^T + A\bar{P} -\tau B_2B_2^T  &\bar{P}C_1^T \\ C_1 \bar{P} & -I\end{bmatrix}<0,
\end{equation}
\begin{equation}\label{neticdxqzo2x}
\begin{bmatrix}
\bar{P} & C_2Q \\ QC_2^T & W\end{bmatrix}>0,
\end{equation}
\begin{equation}\label{neticdxqzo2z}
\mathrm{tr}(W)+\mathrm{tr}(C_1QC_1^T)<\frac{\gamma_2^2}{N-1}.
\end{equation}
Then, the protocol \eqref{comdiso1}  with $G=QC_2^T$, $F=-cB_2^T\bar{P}^{-1}$  and $c\geq \frac{\tau}{2\lambda_2}$ { achieves $H_2$  consensus.}
\end{theorem}

\begin{proof}
In light of Theorem \ref{h2c},  { the network \eqref{netcd1}    achieves $H_2$  consensus}  if the $N-1$ subsystems in \eqref{neticd} are internally stable and $\|\tilde{T}_{\tilde{w}_{0i} \tilde{z}_i}\|_2^2<\frac{\gamma_2^2}{N-1}$.
According to Lemma \ref{lemtr}, the $i$-th subsystem in \eqref{neticd} is internally stable and $\|\tilde{T}_{\tilde{w}_{0i} \tilde{z}_i}\|_2^2<\frac{\gamma_2^2}{N-1}$, if there exist  $Q>0$ satisfying \eqref{neticdxqzo} and   $P>0$ such that
 \begin{equation}\label{neticdxqzo4}
(A+\lambda_iB_2F)^TP+P(A+\lambda_i B_2F)+C_1^TC_1<0,
\end{equation}
and
\begin{equation}\label{neticdxqzo5}
\mathrm{tr}(C_2QPQC_2^T)+\mathrm{tr}(C_1QC_1^T)<\frac{\gamma_2^2}{N-1}.
\end{equation}
Let $\bar{P}=P^{-1}$. Multiplying on both sides of \eqref{neticdxqzo4} by $\bar{P}$ and in light of Schur Complement Lemma \cite{boyd1994linear}, we  obtain that \eqref{neticdxqzo4}  and \eqref{neticdxqzo5}  hold if and only if
 \begin{equation}\label{neticdxqzo6}
\begin{bmatrix} \bar{P}(A+\lambda_iB_2F)^T+(A+\lambda_i B_2F)\bar{P} &\bar{P}C_1^T \\ C_1 \bar{P} & -I\end{bmatrix}<0,
\end{equation}
and
the inequalities \eqref{neticdxqzo2x} and \eqref{neticdxqzo2z} hold at the same  time.
Evidently,
if we choose $F=-cB_2^T\bar{P}^{-1}$ and $c\geq \frac{\tau}{2\lambda_2}$, then \eqref{neticdxqzo2} implies \eqref{neticdxqzo6} and thereby \eqref{neticdxqzo4}.
\end{proof}


\begin{remark}
The separation property of observed-based controllers shown in \cite{JIAO2020109164,haesaert2018separation} is employed in this theorem. The observer gain matrix $ G$ and the feedback gain $F$ are designed in a decoupled way.  Moreover, the feasibility of \eqref{neticdxqzo2} is equivalent to that of \eqref{neticdxqzo6}.
Note that by letting $F\bar{P}=V$ and $\lambda_i=1$, we know that \eqref{neticdxqzo6} holds, then
$$
\bar{P}A+A^T\bar{P}+B_2V+V^TB_2^T+\bar{P}C_1^TC_1\bar{P}<0,$$
which, in light of Finsler's Lemma \cite{iwasaki1994all,li2014cooperative}, is equivalent to that there exist $\bar{P}>0$ and $\tau>0$ such that \eqref{neticdxqzo2} holds. Therefore,
\eqref{neticdxqzo6} implies \eqref{neticdxqzo2}. The converse was shown in the proof. \end{remark}

\subsubsection{Step Two}

In the second step, we design an additional regulating control input $u_{\infty i}$ to deal with the external  disturbances $w_i$ and to guarantee the $H_\infty$ robustness  while not significantly compromising the $H_2$ performance. Since the  noise $w_{0i}$ has been taken care of and filtered out in the first  step,  we only consider the effect of $w_i$ in the second step.

Under the $H_2$ consensus protocol \eqref{comdiso1} in the first step, the augmented agent dynamics are described by
\begin{equation}\label{comdiso4}
\begin{aligned}
\dot{x}_i &=Ax_i+B_2u_i+B_1w_i,\\
\dot{v}_i &=(A-GC_2)v_i+\sum_{j=1}^Na_{ij}\Big(  B_2F(v_{i}-v_{j})\\
&\qquad+GC_2(x_{i}-x_{j})+GD_1(w_i-w_j)\Big),\\
u_i &=u_{2i}+ u_{\infty i},\\
u_{2i} &=Fv_i,\quad i=1,\cdots,N,
\end{aligned}
\end{equation}
where the gain matrices $F$ and $G$  are designed in the first step.

The residual signal $f=[f_1^T,\cdots,f_N^T]^T$ in Fig. \ref{fig2} is used in the second step to activate the {\em inner loop}. It   builds on the protocol \eqref{comdiso1} and is given by
\begin{equation}\label{comdiso5}
\begin{aligned}
f_i &\triangleq  C_2v_i-\sum_{j=1}^Na_{ij}(y_{i}-y_{j}) \\
&= C_2 e_i -D_1\sum_{j=1}^Na_{ij}(w_{i}-w_{j}),~ i=1,\cdots,N,
\end{aligned}
\end{equation}
where $e_i$ is defined as in \eqref{error_e}.

In this step, we consider a distributed protocol of the form
\begin{equation}\label{comdiso6}
    \begin{aligned}
    \dot{n}_i & =  A_cn_i+ B_cf_i, \\
    u_{\infty i} & =  C_cn_i+D_cf_i,
    \end{aligned}
\end{equation}
where $n_i\in\mathbf{R}^{n}$ is the state of the protocol,  and $A_c$, $B_c$, $C_c$, $D_c$ are protocol matrices to be designed. In this case, $Q(s)=\left[
\begin{array}{c|c}A_c & B_c\\ \hline C_c & D_c\end{array}\right]$ in Fig. \ref{fig2}. It should be mentioned that here we assume that $u_{\infty i}$ is a general dynamic controller with $f_i$ as its input. Special forms such as observer-based ones in \cite{trentelman2013robust,lixw2018robust} can be also considered.

Note that the error $e_i$ in the current case satisfies
\begin{equation}\label{comdiso2x}
\begin{aligned}
\dot{e}_i=(A-GC_2)e_i+(GD_1-B_1)\sum_{j=1}^Na_{ij}(w_{i}-w_{j}).
\end{aligned}
\end{equation}
Evidently, if external disturbances $w_i$ are equal to zero, then $e_i$ and thereby $f_i$ will asymptotically converge to zero. In other words, the inner loop will be activated as $f_i$ are bounded signals and $u_{\infty i}$ will be implemented to recover robustness.

Denote $e=[e_1^T,\cdots,e_N^T]^T$, $n=[n_1^T,\cdots,n_N^T]^T$ and $\zeta =\begin{bmatrix} x^T & e^T\end{bmatrix}^T$.
Using \eqref{comdiso4} and \eqref{comdiso6}, we obtain the closed-loop network dynamics in compact form as
\begin{equation}\label{comdiso7}
\begin{aligned}
    \begin{bmatrix}\dot{\zeta} \\\dot{n}\end{bmatrix}&=\begin{bmatrix}\bar{\mathscr{A}}+\mathscr{B}_2\mathscr{D}_c\mathscr{C}_2 & \mathscr{B}_2\mathscr{C}_c\\
    \mathscr{B}_c\mathscr{C}_2 &\mathscr{A}_c\end{bmatrix}\begin{bmatrix} \zeta\\n\end{bmatrix}
    +\begin{bmatrix}\mathscr{B}_1-\mathscr{B}_2\mathscr{D}_c\mathscr{D}_1\\ -\mathscr{B}_c\mathscr{D}_1\end{bmatrix}w, \\
    z & = \begin{bmatrix}\mathscr{C}_1 & 0\end{bmatrix}\begin{bmatrix} \zeta\\n\end{bmatrix},
\end{aligned}
\end{equation}
where $\mathscr{C}_1$ is defined in \eqref{netcd10}, and
\begin{equation}\label{comdiso70}
\begin{aligned}
\bar{\mathscr{A}} &=\begin{bmatrix} I\otimes A  +\mathcal {L}\otimes B_2 F & I\otimes B_2F \\ 0 & I\otimes (A-GC_2)\end{bmatrix},\\
\mathscr{B}_2 &=\begin{bmatrix} I\otimes B_2  \\ 0 \end{bmatrix},~ \mathscr{B}_1=\begin{bmatrix} I\otimes B_1 \\ \mathcal {L}\otimes (GD_1-B_1)\end{bmatrix},\\
\mathscr{C}_2 &=I\otimes \begin{bmatrix} 0 & C_2\end{bmatrix},~\mathscr{D}_1=  \mathcal {L}\otimes  D_1, ~\mathscr{A}_c =I\otimes A_c,\\
\mathscr{B}_c &=I\otimes B_c,~\mathscr{C}_c =I \otimes C_c,~\mathscr{D}_c =I \otimes D_c.
\end{aligned}
\end{equation}

By following similar steps in deriving Theorem \ref{h2c}, it is not difficult to obtain the following result.

\begin{theorem}  \label{hinfc}
Assume that graph $\mathcal {G}$ is connected. Let $\gamma_\infty >0$.
Then, the network \eqref{comdiso7}  { achieves $H_\infty$ consensus} if and only if the following $N-1$ subsystems:
\begin{equation}\label{neticdcd}
\begin{aligned}
\begin{bmatrix}\dot{\tilde{\zeta}}_i \\ \dot{\tilde{n}}_i\end{bmatrix}&=
   (\mathfrak{A}_i+\mathfrak{B}_2\mathfrak{K}\mathfrak{C}_2)\begin{bmatrix}\tilde{\zeta}_i \\ \tilde{n}_i\end{bmatrix}+( \mathfrak{B}_{1i}+\mathfrak{B}_2\mathfrak{K} \mathfrak{D}_{1i})\tilde{w}_{i}, \\
    \tilde{z}_i & = \mathfrak{C}_{1}\begin{bmatrix}\tilde{\xi}_i \\\tilde{n}_i\end{bmatrix},~i=2,\cdots,N,
\end{aligned}
\end{equation}
are internally stable  and the associated transfer functions satisfy $\|\tilde{T}_{\tilde{w}_{i} \tilde{z}_i}\|_\infty<\gamma_\infty$,
where
\begin{equation*}\label{neticdcd0}
\begin{aligned}
 \mathfrak{A}_i&=\begin{bmatrix} A  +\lambda_i B_2 F & B_2F & 0 \\ 0 & A-GC_2 & 0\\ 0 & 0 & 0\end{bmatrix},~\mathfrak{B}_2=\begin{bmatrix} 0 & B_2\\ 0& 0\\I & 0\end{bmatrix},\\
\mathfrak{K}&=\begin{bmatrix} A_c & B_c\\ C_c& D_c\end{bmatrix},~
\mathfrak{C}_2=\begin{bmatrix} 0 & 0 & I \\ 0 &C_2& 0\end{bmatrix},~
\mathfrak{C}_{1}^T=\begin{bmatrix} C_1^T  \\ 0 \\ 0\end{bmatrix},\\
\mathfrak{B}_{1i}&=\begin{bmatrix} B_1\\  \lambda_i(GD_1-B_1) \\ 0\end{bmatrix},~
\mathfrak{D}_{1i}=\begin{bmatrix} 0   \\ -\lambda_i D_1\end{bmatrix}.
\end{aligned}
\end{equation*}
\end{theorem}

The following theorem  provides a design method for obtaining the distributed control law \eqref{comdiso6}.

\begin{theorem} \label{thhinfo}
Assume that $\mathcal {G}$ is connected and that $B_2$ is of full column rank. Let $\gamma_\infty >0$.
Then the network \eqref{comdiso7} { achieves $H_\infty$ consensus} if there exist   positive definite matrices $S_{11}$ and $S_{22}$, and a matrix $Q_1$ such that
\begin{equation}\label{SandV}
	S = \begin{bmatrix}
		S_{11} & 0 \\
		0 & S_{22}
	\end{bmatrix},\qquad
V = \begin{bmatrix}
	V_1 \\
	0
\end{bmatrix}
\end{equation}
satisfying the following LMI's
\begin{equation}\label{comdiso7a2}
\begin{bmatrix}
	\Upsilon_{1i} & S \bar{\mathfrak{B}}_{1i} + V  {\mathfrak{D}}_{1i} \\
(S \bar{\mathfrak{B}}_{1i} + V  {\mathfrak{D}}_{1i})^T & -\gamma_\infty^2 I
\end{bmatrix}
 <0,
 \end{equation}
 for $i=2,N$, where
 \begin{align*}
&  \Upsilon_{1i}=  \bar{\mathfrak{A}}_{i}^T S + S\bar{\mathfrak{A}}_{i} + \bar{\mathfrak{C}}_{2}^T V^T + V\bar{\mathfrak{C}}_{2} + \bar{\mathfrak{C}}_{1}^T \bar{\mathfrak{C}}_{1},\\
&\bar{\mathfrak{A}}_{i}  = T  \mathfrak{A}_{i} T^{-1}, \ \bar{\mathfrak{B}}_{1i} = T \mathfrak{B}_{1i},\
\bar{\mathfrak{C}}_{2} = \mathfrak{C}_{2}T^{-1},\
\bar{\mathfrak{C}}_{1} =  \mathfrak{C}_{1}  T^{-1},
\end{align*}
and $T$ is a nonsingular matrix such that $T\mathfrak{B}_2=\begin{bmatrix} I \\ 0\end{bmatrix}$.
Then, the system matrix $\mathfrak{K}$ of \eqref{comdiso6} is given by
\begin{equation}\label{comdiso7a3}
\mathfrak{K}=S_{11}^{-1}V_1.
\end{equation}
\end{theorem}

\begin{proof}
In virtue of Theorem \ref{hinfc} and Lemma \ref{brl}, it follows that the $N-1$ subsystems in \eqref{neticdcd} are internally stable and $\|\tilde{T}_{\tilde{w}_{i} \tilde{z}_i}\|_\infty<\gamma_\infty$ if and only if there exist matrices $S_i>0$  such that
\begin{equation}\label{comdiso7ap1}
\begin{aligned}
&	(\mathfrak{A}_i+\mathfrak{B}_2\mathfrak{K}\mathfrak{C}_2)^T S_i + S_i (\mathfrak{A}_i+\mathfrak{B}_2\mathfrak{K}\mathfrak{C}_2)\\
&\quad + \frac{1}{\gamma_\infty^2} S_i (\mathfrak{B}_{1i}+\mathfrak{B}_2\mathfrak{K} \mathfrak{D}_{1i}) (\mathfrak{B}_{1i}+\mathfrak{B}_2\mathfrak{K} \mathfrak{D}_{1i})^T S_i \\
&\quad  + \mathfrak{C}_{1}^T \mathfrak{C}_{1}<0, \quad i=2,\cdots,N.
\end{aligned}
\end{equation}
Following the steps in \cite[Theorem 2]{Jiayingmin2009acc}, by using a Schur complement, the above inequalities \eqref{comdiso7ap1} are equivalent to
\begin{equation} \label{midstep}
\begin{bmatrix}
\Phi_{1i} & \Phi_{2i} \\
\Phi_{2i}^T & -\gamma_\infty^2 I
\end{bmatrix}<0, \quad i=2,\cdots,N
\end{equation}
with
\begin{align*}
\Phi_{1i} & = (\mathfrak{A}_i+\mathfrak{B}_2\mathfrak{K}\mathfrak{C}_2)^T S_i + S_i (\mathfrak{A}_i+\mathfrak{B}_2\mathfrak{K}\mathfrak{C}_2)+ \mathfrak{C}_{1}^T \mathfrak{C}_{1},\\
\Phi_{2i} &= S_i (\mathfrak{B}_{1i}+\mathfrak{B}_2\mathfrak{K} \mathfrak{D}_{1i}).
\end{align*}
Since the matrix $B_2$ is of full column rank, there exists a matrix $T$ such that  $$T\mathfrak{B}_2=\begin{bmatrix} I \\ 0\end{bmatrix}.$$
 By pre-multiplying
$
 \bar{T} = \begin{bmatrix}
T^{-T} & 0 \\ 0 & I
 \end{bmatrix}
$ and post-multiplying $\bar{T}^T$
on \eqref{midstep}, it follows that \eqref{midstep} holds if and only if
\begin{equation} \label{midstep2}
	\begin{bmatrix}
		\bar{\Phi}_{1i} & \bar{\Phi}_{2i} \\
		\bar{\Phi}_{2i}^T & -\gamma_\infty^2 I
	\end{bmatrix}<0,
 \quad i= 2,\ldots,N,
\end{equation}
where
\begin{align*}
	& \bar{\Phi}_{1i} =  \bar{\mathfrak{A}}_{i}^T \bar{S}_i + \bar{S}_i \bar{\mathfrak{A}}_{i} +   \bar{S}_i \bar{\mathfrak{B}}_{2} \mathfrak{K} \bar{\mathfrak{C}}_{2} + (\bar{S}_i \bar{\mathfrak{B}}_{2} \mathfrak{K} \bar{\mathfrak{C}}_{2})^T +\bar{\mathfrak{C}}_{1}^T \bar{\mathfrak{C}}_{1}, \\
& \bar{\Phi}_{2i}= \bar{S}_i \bar{\mathfrak{B}}_{1i} + \bar{S}_i \bar{\mathfrak{B}}_{2}  \mathfrak{K} {\mathfrak{D}}_{1i},\
 \bar{S}_i = T^{-T} S_i T^{-1}, \\
 &  \bar{\mathfrak{B}}_{2} = T  \mathfrak{B}_{2} =\begin{bmatrix} I & 0\end{bmatrix}^T.
\end{align*}

Now, let $S = \bar{S}_i$ and $\bar{S}_i \bar{\mathfrak{B}}_{2}  \mathfrak{K} = V$. Recall that \eqref{SandV}, then \eqref{midstep2} holds if \eqref{comdiso7a2} holds for $i=2, \ldots,N$. In this case, due to
\begin{equation*}
\begin{bmatrix}
	S_{11} & 0 \\
	0 & S_{22}
\end{bmatrix}
\begin{bmatrix}
	 I \\ 0
 \end{bmatrix}\mathfrak{K}
=
\begin{bmatrix}
	V_1 \\
	0
\end{bmatrix},
\end{equation*}
the system matrix is given by \eqref{comdiso7a3}.

Finally, note that the inequalities in \eqref{comdiso7a2}  are linear matrix inequalities with respect to the unknown variables, we need to check only the two LMIs in \eqref{comdiso7a2} for $i=2$ and $N$, and the other $N-3$ LMIs in \eqref{comdiso7a2} corresponding to $i=3,\cdots,N-1$, also hold, with their variables chosen to be some convex combinations of those satisfying the two LMIs in \eqref{comdiso7a2} for $i=2,N$. This completes the proof.

\end{proof}

\begin{remark}
In the novel control structure in Fig. 1, the design of $H_2$ consensus control of the outer loop is independent of the control $Q(s)$ in \eqref{comdiso6} of the inner loop. The extra control $Q(s)$ relies on the residual signal $f$, which is the stacked error between the actual relative outputs of neighboring agents  $\sum_{j=1}^Na_{ij}(y_{i}-y_{j})$ and their observed ones $C_2v_i$ given by the distributed observer \eqref{comdiso1a}. The extra control action of the inner loop, activated by residual signal $f$ in the presence of external disturbances or uncertainties $w_i$, will complement the $H_2$ performance by providing { $H_\infty$ robustness} guarantee with respect to $w_i$. This is the reason why this approach is called a complementary design approach.
\end{remark}

\begin{remark}
The two-step complementary  approach proposed in this section, compared to the trade-off approach, has at least two main advantages:

\begin{itemize}
\item[i)] The extra control $Q(s)$ in \eqref{comdiso6} provides an additional degree of freedom for design. Therefore, the current complementary approach has two separate control inputs to deal with the $H_2$ performance and the $H_\infty$ robustness of consensus, respectively, thereby does not need to make {much} trade-off, and can be expected to be much less conservative,while the trade-off design  has only one control to tackle two conflicting performances simultaneously.
{ Although} Theorems \ref{h2cdesign} and \ref{thhinfo} in this section are conservative, however, it should be pointed out that the conservatism of these results is not caused by the complementary approach. On the contrary, it in fact highlights the difficulty of distributed control.

\item[ii)]  The  control action of the inner loop is proportional to the residual signal which quantifies the modeling mismatch level, thereby having some online ``adaptivity" with respect to modeling errors.
This complementary approach will yield   the same achievable $H_2$ performance when modeling mismatches do not exist, because in this case the inner loop will be de-activated.  By contrast, the trade-off approach always considers the {\it a priori} worst case and produce the same conservative performance even when disturbances or uncertainties disappear.
\end{itemize}
\end{remark}

\begin{remark}
It should be mentioned that the order of  the overall control law designed by the complementary approach is higher than the trade-off approach, since both distributed protocols \eqref{comdiso1a} and  \eqref{comdiso6} are required in the former approach while only one dynamic controller is needed in the latter.  This is the price to provide more degree of design freedom and it is actually not a big issue considering the abundance of cheep storage and computing resources.
Besides, the extra control action of  the inner loop, providing robustness guarantee, will inject some  white noises into the closed-loop network dynamics, and thereby will make certain compromise of the $H_2$ performance.
\end{remark}

\subsection{ Absolute Output Feedback Case}

In this subsection, we consider the case where  {\em absolute output} information of each  agent is available.  In this case, we  adopt local observers for the agents, instead of the distributed observers as in the previous subsection.
The structure of the complementary design  approach  in  this case is depicted in Fig.  \ref{fig3}.

\begin{figure}[htbp]
\centering
\includegraphics[width=0.8\columnwidth]{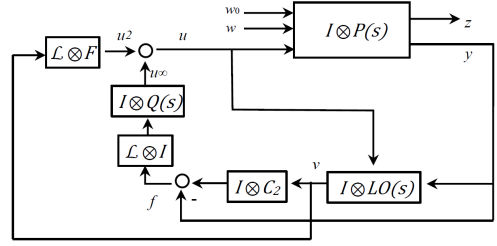}
\caption{The controller structure of the complementary  approach based on absolute outputs, where $LO(s)$ represents the local observer for each agent, $\tilde{f}$ is the residual signal,  and $\tilde{Q}(s)$ is the additional  controller based on $\tilde{f}$, and the rest of variables are defined as in Fig. \ref{fig2}.}\label{fig3}
\end{figure}

\subsubsection{Step One}
Based on the absolute output $y_i$,  we propose for each agent the following Luenberger observer:
\begin{equation}\label{locob}
\dot{\breve{v}}_i=A\breve{v}_i+Bu_i+\breve{G}(y_i-C_2\breve{v}_i),
\end{equation}
where $\breve{G}$ is the observer gain  to be designed. In the first step, we consider only the outer loop,  therefore, $u_i=u_{2i}$. For the  $H_2$  consensus problem, we design the  following protocol:
\begin{equation}\label{locob1}
u_{2i}= \breve{F} \sum_{j=1}^Na_{ij}(\breve{v}_i-\breve{v}_j),
\end{equation}
where $\breve{F}$  is the feedback gain to be designed.
Denote $\breve{v}_i=[\breve{v}_1^T,\cdots,\breve{v}_N^T]^T$ and $\breve{\xi}=\begin{bmatrix} x^T &\breve{v}^T\end{bmatrix}^T$. The closed-loop network dynamics in this case can be written in compact form as
\begin{equation}\label{netcdloc1}
\begin{aligned}
    \dot{\breve{\xi}}&=\begin{bmatrix} I\otimes  A  &  \mathcal {L}\otimes B_2 \breve{F} \\ I\otimes \breve{G}C_2 & I\otimes(A-\breve{G}C_2)+\mathcal {L}\otimes B_2 \breve{F} \end{bmatrix}\breve{\xi} + \begin{bmatrix} I\otimes B_0 \\  I\otimes\breve{G}D_0\end{bmatrix}w_0, \\
    z & = \mathcal {M}\otimes \begin{bmatrix} C_1 & 0\end{bmatrix}\breve{\xi}.
\end{aligned}
\end{equation}

It is easy  to verify that the $H_2$ consensus problem of \eqref{netcdloc1} can be reduced to the same condition as in Theorem \ref{h2c}.  Then, Theorem \ref{h2cdesign} can also be  used  to design the protocol \eqref{locob1}.

\subsubsection{Step Two}

In the second step, we define the residual signals $\tilde{f}_i$ as follows:
\begin{equation}\label{comdiso5x}
\begin{aligned}
\tilde{f}_i &\triangleq  C_2\breve{v}_i-y_{i},~ i=1,\cdots,N,
\end{aligned}
\end{equation}
which is actually the local estimated output error.
Therefore, $\tilde{f}_i$ quantifies the difference between the actual plant and the ideal plant, and we can see that $\tilde{f}_i=0$, if there exist no disturbances or uncertainties.

Since $\sum_{j=1}^Na_{ij}(x_i-x_j)$ is the consensus error and $\breve{v}_i$ is the estimate of $x_i$  for agent $i$, we can see that $\sum_{j=1}^Na_{ij}(\tilde{f}_i-\tilde{f}_j)$ denotes the estimated output error of the consensus error. Therefore, in this case we design the control input $u_{\infty i}$ for the inner loop based on $\sum_{j=1}^Na_{ij}(\tilde{f}_i-\tilde{f}_j)$, instead of $\tilde{f}_i$ as in the previous subsection. Specifically, we consider  a distributed protocol of the form
\begin{equation}\label{comdiso_inf}
	\begin{aligned}
		\dot{n}_i & =  \breve{A}_cn_i+ \breve{B}_c \sum_{j=1}^Na_{ij}(\tilde{f}_i-\tilde{f}_j), \\
		u_{\infty i} & =  \breve{C}_cn_i+\breve{D}_c \sum_{j=1}^Na_{ij}(\tilde{f}_i-\tilde{f}_j),
	\end{aligned}
\end{equation}
where $n_i\in\mathbf{R}^{n}$ is the state of the protocol, {and $\breve{A}_c$, $\breve{B}_c$, $\breve{C}_c$, $\breve{D}_c$ are protocol matrices to be designed}.  In this case,  ${Q}(s)=\left[  \begin{array}{c|c}\breve{A}_c & \breve{B}_c\\ \hline \breve{C}_c & \breve{D}_c\end{array}\right]$ in Fig. \ref{fig3}.

Let $\breve{e}_i = \breve{v}_i-x_i$. Denote $\breve{e}=[\breve{e}_1^T,\cdots,\breve{e}_N^T]^T$,  $\breve{\zeta}=\begin{bmatrix} x^T & \breve{e}_i^T\end{bmatrix}^T$  and
	  $\breve{n}=[\breve{n}_1^T,\cdots,\breve{n}_N^T]^T$. Then, it  follows from \eqref{d1}, \eqref{locob1} and \eqref{comdiso_inf} that the closed-loop network dynamics are given by
\begin{equation}\label{comdisoloc7}
\begin{aligned}
    \begin{bmatrix}\dot{\breve{\zeta}} \\  \dot{\breve{n}}\end{bmatrix}&=\begin{bmatrix}\breve{\mathscr{A}}+\mathscr{B}_2 \breve{\mathscr{D}}_c\mathscr{C}_2 & \mathscr{B}_2 \breve{\mathscr{C}}_c\\
     \breve{\mathscr{B}}_c\mathscr{C}_2 & \breve{\mathscr{A}}_c\end{bmatrix}\begin{bmatrix} \breve{\zeta}\\ \breve{n}\end{bmatrix}
    +\begin{bmatrix} \breve{\mathscr{B}}_1- \mathscr{B}_2 \breve{\mathscr{D}}_c\breve{\mathscr{D}}_1\\ - \breve{\mathscr{B}}_c \breve{\mathscr{D}}_1\end{bmatrix}w, \\
    z & = \begin{bmatrix}\mathscr{C}_1 & 0\end{bmatrix}\begin{bmatrix} \breve{\zeta} \\ \breve{n}\end{bmatrix},
\end{aligned}
\end{equation}
where $\mathscr{C}_1$, $\mathscr{C}_2$, and $\mathscr{B}_2$ are defined in \eqref{netcd10}, and
\begin{equation*}\label{comdisoloc70}
\begin{aligned}
\breve{\mathscr{A}} &=\begin{bmatrix} I\otimes A  +\mathcal {L}\otimes B_2 \breve{F} & \mathcal {L}\otimes B_2 \breve{F} \\ 0 & I\otimes (A- \breve{G}C_2)\end{bmatrix},\\
 \breve{\mathscr{B}}_1 &=
 \begin{bmatrix}
 	I\otimes B_1 \\
 	\mathcal {L}\otimes ( \breve{G}D_1-B_1)
 \end{bmatrix},\
\breve{\mathscr{D}}_1=  I\otimes  D_1, \
\breve{\mathscr{A}}_c  =I\otimes  \breve{A}_c,\\
\breve{\mathscr{B}}_c &=I\otimes  \breve{B}_c,~\breve{\mathscr{C}}_c =I \otimes  \breve{C}_c,~\breve{\mathscr{D}}_c =I \otimes  \breve{D}_c.
\end{aligned}
\end{equation*}

Similarly as in Theorems \ref{hinfc} and \ref{thhinfo}, the control law \eqref{comdiso_inf} can be constructed to achieve $H_\infty$ consensus with prescribed index. The details are omitted here for conciseness.

\section{Simulation Example}\label{sec_simu}
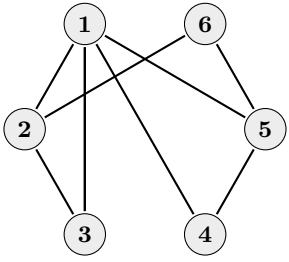
\begin{figure}[t]
	\centering
	\begin{tikzpicture}[scale=0.4]
		\tikzset{VertexStyle1/.style = {shape = circle,
				color=black,
				fill=white!93!black,
				minimum size=0.5cm,
				text = black,
				inner sep = 2pt,
				outer sep = 1pt,
				minimum size = 0.55cm},
			VertexStyle2/.style = {shape = circle,
				color=black,
				fill=black!53!white,
				minimum size=0.5cm,
				text = white,
				inner sep = 2pt,
				outer sep = 1pt,
				minimum size = 0.55cm}
		}
		\node[VertexStyle1,draw](1) at (-2,3.5) {$\bf 1$};
		\node[VertexStyle1,draw](2) at (-4,0) {$\bf 2$};
		\node[VertexStyle1,draw](3) at (-2,-3.5) {$\bf 3$};
		\node[VertexStyle1,draw](4) at (2,-3.5) {$\bf 4$};
		\node[VertexStyle1,draw](5) at (4,0) {$\bf 5$};
		\node[VertexStyle1,draw](6) at (2,3.5) {$\bf 6$};
		\Edge[ style = {-,> = latex',pos = 0.2},color=black, labelstyle={inner sep=0pt}](1)(2);
		\Edge[ style = {-,> = latex',pos = 0.2},color=black, labelstyle={inner sep=0pt}](1)(3);
		\Edge[ style = {-,> = latex',pos = 0.2},color=black, labelstyle={inner sep=0pt}](1)(4);
		\Edge[ style = {-,> = latex',pos = 0.2},color=black, labelstyle={inner sep=0pt}](1)(5);		%
		\Edge[ style = {-,> = latex',pos = 0.2},color=black, labelstyle={inner sep=0pt}](2)(3);	\Edge[ style = {-,> = latex',pos = 0.2},color=black, labelstyle={inner sep=0pt}](2)(6);
		\Edge[ style = {-,> = latex',pos = 0.2},color=black, labelstyle={inner sep=0pt}](3)(1);
		\Edge[ style = {-,> = latex',pos = 0.2},color=black, labelstyle={inner sep=0pt}](4)(5);
		\Edge[ style = {-,> = latex',pos = 0.2},color=black, labelstyle={inner sep=0pt}](5)(6);	
	\end{tikzpicture}
	\caption{The   communication graph among the six agents.}
	\label{graph}
\end{figure}

 In this section we will use a simulation example to illustrate the proposed complementary $H_2$ and $H_\infty$ design by dynamic output feedback,  as in Theorems \ref{h2cdesign} and \ref{thhinfo} in Subsection \ref{subsec_output_feedback}, for obtaining distributed protocols.
 	
Consider a multi-agent system that consists of six agents. The dynamics of each agent is given by \eqref{d1},  where
\begin{equation*}\label{agent_matrices}
	\begin{aligned}
	&	A = \begin{bmatrix}
		-2 & 2 \\-1 & 1
	\end{bmatrix},\
	B_0 = \begin{bmatrix}
	0 & 0 \\ 0.5 & 0
\end{bmatrix}, \
	B_1 = \begin{bmatrix}
		1 \\0.6
	\end{bmatrix},\	B_2 =
\begin{bmatrix}
	1 \\ -2
\end{bmatrix},\\
& 	C_1 = \begin{bmatrix}
	1 & 0
	\end{bmatrix},\
	C_2 = \begin{bmatrix}
		1 & 0.8
	\end{bmatrix}, \  D_1 = 0.1,\
	D_0 = \begin{bmatrix}
		0 & 1
	\end{bmatrix}.
\end{aligned}
\end{equation*}
The communication graph among the agents is shown in Fig. \ref{graph}, which is a connected undirected graph with the Laplacian matrix  $\mathcal{L}$.
The smallest nonzero and the largest eigenvalues of $\mathcal{L}$ are $\lambda_2 = 1.3820$ and $\lambda_N = 5.3028$.

\begin{figure}[t]
	\centering
	\includegraphics[width=\columnwidth]{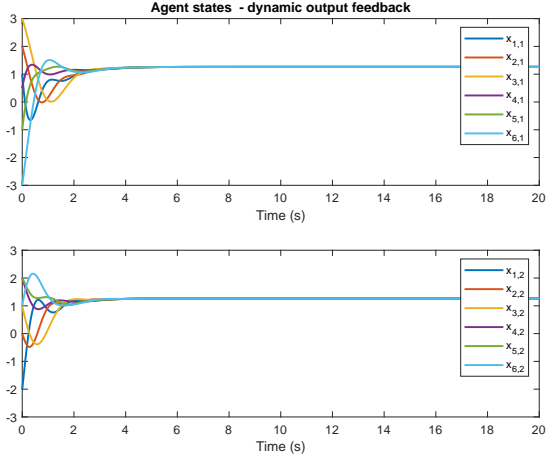}
	\caption{Complementary $H_2$ and $H_\infty$  design by output feedback: plots of the agent state vector $x^1 = (x_{1,1},\ldots, x_{6,1})$ (upper plot) and $x^2 = (x_{1,2},\ldots, x_{6,2})$ (lower plot).}
	\label{com_output_feedback_state_plots}
\end{figure}

\begin{figure}[t]
	\centering
	\includegraphics[width=\columnwidth]{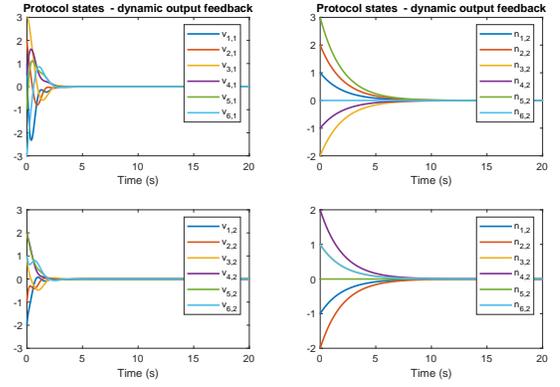}
	\caption{Complementary $H_2$ and $H_\infty$  design by output feedback: plots of the protocol state vector $v^1 = (v_{1,1},\ldots, v_{6,1})$ (upper left plot), $v^2 = (v_{1,2},\ldots, v_{6,2})$ (lower left plot), $n^1 = (n_{1,1},\ldots, n_{6,1})$ (upper right plot) and $n^2 = (n_{1,2},\ldots, n_{6,2})$ (lower right plot).}
	\label{com_output_feedback_protocol_state_plots}
\end{figure}
%

\begin{figure}[t]
	\centering
	\includegraphics[width=\columnwidth]{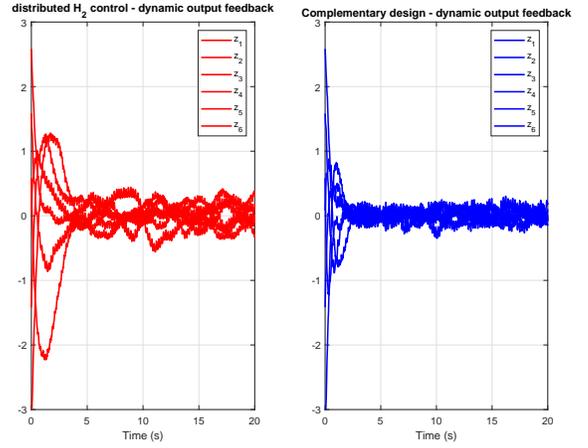}
	\caption{Complementary $H_2$ and $H_\infty$  design by output feedback: plots of the agent output vector $z^1 = (z_{1},\ldots, z_{6})$.}
	\label{com_output_feedback_output_plots_noisy}
\end{figure}

In the first step of the complementary design, following Theorem \ref{h2cdesign},  we obtain a distributed control law   that takes care of the $H_2$ performance. We choose $\gamma_2 = 2$ and compute the control gains
$
	F = \begin{bmatrix}
		-0.1627 &   0.7430
	\end{bmatrix}$ and $
	G = \begin{bmatrix}
0.5685 &
0.7966
	\end{bmatrix}^T$.
Next in the second step, following Theorem \ref{thhinfo}, we obtain a distributed control law that deals with the  $H_\infty$ consensus. We compute the   protocol gains to be
$
	A_c   =
		\begin{bmatrix}
		-0.5031 &  0 \\
		0   & -0.5031
	\end{bmatrix},$
$B_c = \begin{bmatrix}
	0 \\ 0
\end{bmatrix},$
$C_c  = \begin{bmatrix}
	0 & 0
\end{bmatrix}$ and
$D_c =
-0.4045.$

The associated computed  upper bound for the  $H_\infty$ robustness is   $\gamma_{\infty, {\rm min}} =1.5808$.
In Fig. \ref{com_output_feedback_state_plots}, we have plotted the state trajectories of the agents, and in Fig.  \ref{com_output_feedback_protocol_state_plots} we have plotted the state trajectories of the two proposed distributed protocols. It can be seen that indeed the proposed distributed protocols together achieve consensus for the multi-agent system.

As a comparison, we will next compare, in presence of the external noise and disturbance, the output performance of the proposed complementary approach with that of the distributed $H_2$ control.  In particular, we choose  the external noise $w_{0i}$  to be a uniformly distributed signal, generated by Matlab command \texttt{30*rand()}. We choose the  disturbance  $w_{i}$ to be
$
 	w_{1} =   w_{3} =  3\sin (110t),\ w_{2} = w_4 =  3\sin (30t),  w_{5} =  w_{6} = 3\sin (60t).$
The plots of the trajectories of the performance outputs $z_i$  are given in Fig. \ref{com_output_feedback_output_plots_noisy}. It can be seen that, in the presence of noises and disturbances, indeed the proposed complementary  approach guarantees a better performance than the distributed $H_2$ control.

\section{Conclusions}\label{sec_conclusions}

In this paper, we have presented a novel complementary approach  to the distributed $H_2$ and $H_\infty$ consensus problem of   multi-agent systems.
Through introducing an extra control input that depends on some carefully chosen residual signals, which indicates the modeling mismatch, the complementary approach provides an additional degree of freedom for control design and complements the $H_2$ performance of consensus by providing the $H_\infty$ robustness guarantee.
This complementary approach does not involves much trade-off, and can be expected to be much less conservative than the trade-off design.

Future works include extending the proposed complementary approach to other robust optimal cooperative control problems with different optimal performances under different types of uncertainties.

\balance

\end{document}